\documentclass[a4paper,USenglish]{lipics-v2021}
\usepackage{amsthm}
\usepackage[utf8]{inputenc}
\usepackage[dvipsnames]{xcolor}
\usepackage[ruled,vlined]{algorithm2e}
\usepackage{tikz}
\usetikzlibrary{decorations.pathreplacing}
\usepackage{stmaryrd}
\definecolor{colorcolette}{rgb}{0.54, 0.17, 0.89}
\definecolor{colormikael}{rgb}{0.98, 0.3, 0.0}
\definecolor{colortistou}{rgb}{0.1, 0.1, 0.9}

\newcommand{\Tistou}[1]{\noindent{\color{colortistou} {\\Tistou: #1}}}

\newcommand{\p}{\operatorname{p}}
\newcommand{\InAlg}{\mathcal A}
\newcommand{\sleepingDuration}{\operatorname{sleepDuration}}
\newcommand{\countVar}{\operatorname{ect}}
\newcommand{\maxCount}{\operatorname{MSPect}}
\newcommand{\ectMinStrictPositif}{\text{GMSPect}}
\newcommand{\correct}{\operatorname{correct}}
\newcommand{\pseudoCorrect}{\operatorname{pseudoCorrect}}

\title{Introducing the Self-Stabilizing SLEEPING Model}

\author{Tistou Fages}{LIP, ENS de Lyon, France}{TBD}{}{}

\author{Colette Johnen}{LaBRI, CNRS, Bordeaux, France}{colette.johnen@labri.fr}{https://orcid.org/0000-0001-7170-4521}{}

\author{Mikaël Rabie}{Université Paris Cité, IRIF, CNRS, Paris, France}{mikael.rabie@irif.fr}{https://orcid.org/0000-0001-6782-7625}{}

\Copyright{Colette Johnen and Mokaël Rabie} 



\funding{This research was funded, in whole or in part, by the Agence Nationale de la Recherche (ANR), grant ANR-24-CE48-7768-01 (project ENEDISC) and  grant ANR-22-CE48-0001 (projet TEMPOGRAL)}

\hideLIPIcs

\nolinenumbers
  
\authorrunning{Fages, Johnen and Rabie}  
  
\acknowledgements{The authors are thankful to Lélia Blin and Sylvain Gay for the preliminary discussions about the model, as well as Johanne Cohen and Laurence Pilard for the discussions on the algorithms.}

\EventEditors{John Q. Open and Joan R. Access}
\EventNoEds{2}
\EventLongTitle{OPODIS 2026}
\EventShortTitle{OPODIS 2026}
\EventAcronym{OPODIS}
\EventYear{2026}
\EventDate{November 9--13, 2026}
\EventLocation{Rome, Italy}
\EventLogo{}
\SeriesVolume{42}
\ArticleNo{23}

\begin{document}

\maketitle

\begin{abstract}
The SLEEPING LOCAL model introduces a new complexity parameter, the awake complexity, to make distributed algorithms energy-efficient. 
In the synchronous LOCAL model, nodes can now decide to be awake or asleep in each round. 
In a round, only awake nodes can communicate to share information, which consumes energy.
The awake complexity is the maximum number of times a node is activated to produce an output. 
In particular, it often comes at the cost of the total number of rounds required to solve a problem, compared with algorithms in which every node is awake in every round.

In this article, we adapt the notion of awaken rounds to the context of self-stabilization, introducing the Self-Stabilizing SLEEPING model. 
Nodes are no longer required to remain awake at all times. 
However, in self-stabilization, nodes must be activated infinitely often to detect any issue in the system's current state. 
In this model, the complexities are:
\begin{itemize}
    \item How many synchronous rounds are needed to reach a legitimate configuration?
    \item How many times does a node need to be awake to reach this configuration?
    \item How often does a node need to be awake once this configuration is reached?
\end{itemize}
The goal is to minimize those three metrics, and we can expect different trade-offs.

We present energy-efficient algorithms to solve the problems of finding a $(\Delta+1)$-coloring, a Maximal Independent Set, and a Maximal Matching, thanks to new ad hoc sleeping techniques that reduce the awake complexity (i.e., energy consumption) during the convergence phase.

We also propose two transformers that adapt silent self-stabilizing algorithms to the SLEEPING setup. 
The first transformer is pretty simple and deals with low-complexity algorithms. The second is more elaborate and is more energy-efficient when it transforms slow self-stabilizing algorithms.

\end{abstract}


\maketitle
\newpage
\setcounter{page}{1}
\section{Introduction}

\subsection{Motivation}

In a long-lived system that runs forever, one should expect some failures to occur (e.g., memory corruption or network modification). To address these failures, one solution is to run self-stabilizing algorithms that can recover from any problematic configuration. However, detecting one of those failures entails a high energy cost: all nodes must run their detection systems continuously.

In the synchronous LOCAL model, the SLEEPING variant was recently introduced. In this variant, nodes decide when they are actually awake. In each round, only awake nodes are allowed to communicate with their (awake) neighbors. It allows for saving energy and permits nodes to operate only when needed. In particular, we obtain algorithms in which nodes have fewer awake rounds than the total number of rounds in the classical LOCAL model. However, it usually comes at a time cost: the total number of rounds needed before the system reaches its final state is usually considerably larger.

What if we introduce the notion of SLEEPING in Self-Stabilizing algorithms? In particular, it would permit nodes to be non-continuously awake for detecting system failures, thereby saving a considerable amount of energy. On the other hand, we would want the system to be fast enough to recover from a failure once it is detected.

A considerable challenge is to determine what constitutes a good definition of a SLEEPING version of self-stabilizing algorithms. The first reflex is to add, as in the LOCAL version, an internal clock to determine when the node should be activated next. Do we want the internal clock to be public or private? Making it public allows two neighbors to synchronise their clocks. However, it increases the number of public states and requires each node to communicate its updated clock in each round. In this paper, we consider the case where the clock is internal: a node can report how many rounds it will be sleeping, but its neighbors cannot deduce the exact time at which they will be awake again.

Some algorithms operate in an asynchronous scheduler in which a daemon determines which nodes are activated in each round. This becomes challenging in the case of sleeping nodes: should their clock be updated when they are scheduled, or in each round? In the SLEEPING LOCAL model, the advantage of being sleeping is to wait for some event to happen in the meantime. However, in an asynchronous setting, we lose this guarantee. For that reason, we consider a synchronous setting in this article and leave open the question of how to adapt the model to an asynchronous setting.

In the SLEEPING LOCAL model, information can only go through an edge if both endpoints are awake. Several techniques are developed in this model to coordinate activation among neighbors to be energy-efficient while passing information at the appropriate time. In particular, having access to a common clock and starting computation synchronously in round one facilitates coordination. In an asynchronous setting, if both endpoints of an edge have to be awake at the same time to communicate, it would force nodes to be awake at the same time infinitely often to detect an issue. It would impose an additional layer of constraints on the internal clocks. In this article, we consider the state model, in which nodes have access to their neighbors' latest state when they are awake. This allows balancing, given that the public state does not allow nodes to know when their neighbor will next activate (unless the neighbor is awake in every round).

For a classical self-stabilizing algorithm, the usual complexity metric is the number of rounds needed for the system to reach, from any initial configuration, a good configuration. However, unlike the LOCAL model, nodes cannot reach a state in which they are done and stop updating themselves during self-stabilization. One important metric is how often a node is awake over time, particularly when a legitimate configuration is reached.

Note that this frequency directly affects the total number of rounds required for the system to converge to a legitimate configuration. A node can start in this low-frequency state, even though some update must be performed. As a consequence, the later the nodes wake up for the first time, the longer the system takes to detect and broadcast that a configuration update needs to be computed. On the other hand, we want the system to update itself fast enough. As in the SLEEPING LOCAL model, this will introduce interesting trade-offs in complexity.

\subsection{Related Work}

The LOCAL model, introduced in 1987~\cite{linial1987distributive,peleg2000distributed}, aims to capture computation on graphs, in which nodes (or processors) coordinate synchronously to solve a problem locally. In particular, LCL problems~\cite{naor1993can} - problems where a solution can be checked locally, such as coloring, maximal independent set, and maximal matching - are intensively studied. 
Distributed graph coloring is one of the most fundamental and extensively studied problems in the LOCAL model~\cite{BEG22, FYZ25}.
In~\cite{BEG22} is presented a locally-iterative $(\Delta + 1)$-coloring algorithm converging in  $O(\Delta + \log^* n)$ synchronous rounds.
A general challenge is to find the best complexity to solve a given problem with regard to $\Delta$ (the maximal degree of the graph) and $n$ (the size of the network). See~\cite{suomela2013survey} for a survey on algorithms in the LOCAL model.

In 2020, the SLEEPING LOCAL model was introduced in~\cite{chatterjee2020sleeping} to account for energy consumption. 
Nodes decide when to be active (or awake) and when to sleep, with the aim of minimizing the energy consumption (i.e., the number of awake rounds for each node). 
Thanks to the results of~\cite{augustine2024awake,barenboim2021deterministic}, it is known that any global problem has awake complexity $\Theta(\log n)$, compared to $\Theta(n)$ in the LOCAL model. 
In~\cite{BFOR25}, an algorithm is presented to solve $(\Delta+1)$-coloring and Maximal Independent set, deterministically, with an awake complexity of $O(\log^*n\sqrt{\log n})$. 
Randomized algorithms with constant average complexity to compute a Maximal Matching or a Maximal Independent Set are provided in~\cite{dufoulon2023distributed,ghaffari2022average,ghaffari2023distributed}. 
One of the current challenges in this model is studying the awake complexity of the deterministic maximal matching algorithm.
Does the $\Omega(\Delta)$ lower bound in the LOCAL model~\cite{balliu2021lower} hold for the awake complexity?
or an algorithm in $o(\Delta)$ awake rounds can be designed? This is one of the current main open questions in this model.

Energy consumption through choosing when to be awake has been considered in wireless sensor networks. In~\cite{ye2002energy,van2003adaptive,herman2007temporal}, nodes are periodically active. During a round, a node can only communicate (and know about) neighbor nodes that are also active in that round.
In this context, the main challenge is to ensure that each node identifies all of its neighbors and that they coordinate their active and inactive cycles to enable periodic communication in their local neighborhood.
See~\cite{sun2014energy} for a survey.

There are bridges between the LOCAL model, LCL problems, and self-stabilization.
Various self-stabilizing algorithms to solve LCL problems have been proposed (see~\cite{GK10} for a survey). Some are designed for specific graph topologies, such as planar graphs~\cite{GK93} and bipartite graphs~\cite{KK06}. Greedy problems have also been considered~\cite{cohen2023making,HJS03,GT00}.

More generally, any (silent) algorithm in the LOCAL model converging into $T$ synchronous rounds can be transformed into 
a self-stabilizing algorithm that converges  into $3T+2$ rounds
~\cite{LSW09,JDMI24} at the expense of the memory space.

\subsection{Contribution}

We propose three self-stabilizing SLEEPING algorithms that solve well-known distributed problems: $(\Delta+1)$-coloring (Section~\ref{sec:col}), maximal independent set (Section~\ref{sec:mis}), and maximal matching (Section~\ref{sec:mm}). We work in networks where nodes have unique identifiers and share knowledge of an upper bound $M$ on the network size.
Figure~\ref{fig:algs} summarizes these algorithms by presenting the round complexity, the awake complexity before convergence, and the awake frequency after convergence. 
Nodes share an upper bound $M$ on the larger identifier in the network, and $\Delta$ is the maximum degree of the graph.
The round complexity is, of course, important for the three algorithms, but their awake complexities, which depend only on $\Delta$, are close to the step complexity of the best self-stabilizing algorithms without any sleeping nodes~\cite{Turau07, MMPT09, GK10}.

For $(\Delta+1)$-coloring and the maximal independent set problem, a node's wake-up frequency depends on whether it has a conflict. On the other hand, for maximal matching, the frequency remains the same, as a node might need to change its state depending on information at distance 2 (for example, if itself and another node are matched to the same node). The awake frequency before convergence, in the three algorithms, is a unique prime number for each node, depending on their identifier, $\Delta$, and even $M$ for the maximal matching algorithm. Those prime numbers ensure that a mutual exclusion will happen every $(\Delta+1)$ awake round of a node, allowing them to do a step without any interference from one of their neighbor.

The round complexities, which depend on $M$, are based on the Prime Number Theorem~\cite{hadamard1896distribution}, which allows us to state that the $k$-th prime number is $O(k\log k)$.

\begin{figure}[ht!]
    \begin{tabular}{|c|c|c|c|}
      \hline
SS SLEEPING &  & awake & awake  \\
 Algorithm  & round & complexity  & frequency  \\
 & complexity & before & after \\
 &  &  convergence & convergence \\
 \hline \hline
$\Delta+1$ coloring - Algorithm~\ref{alg:coloring}  &
$O(\Delta M\log M)$ & $\Delta+O(1)$ & $O(1/(\Delta M\log M))$ \\
\hline
Maximal independent set - Algorithm~\ref{alg:mis}  &
$O(\Delta M\log M)$ & $2\Delta+O(1)$ & $O(1/(\Delta M\log M))$ \\
\hline
Maximal matching - Algorithm~\ref{alg:mm}  &
$O(\Delta^2 M\log M)$ & $O(\Delta^2)$ & $O(1/(M\log M))$ \\
\hline
\end{tabular}
\caption{Our Self-Stabilizing SLEEPING Algorithms.}
\label{fig:algs}
\end{figure}

We also propose two transformers (Section~\ref{sec:trans}) that build, from
a silent self-stabilizing algorithm, called $\InAlg$,
a silent self-stabilizing SLEEPING algorithm solving the same problem as $\InAlg$. 
In both algorithms, the round complexity depends on $T$, the convergence time of $\InAlg$, and some arbitrary constant $SP$.
The first transformer (that is very simple) requires that $\InAlg$
converges under the weakly fair distributed scheduler.
The second only requires that $\InAlg$ converges under the synchronous scheduler.
Both transformers have the same awake frequency after convergence; however, the convergence time differs. In particular, using the second transformer becomes more interesting if $\InAlg$ has a complexity super-linear in $D$, the diameter of the graph.

 \begin{figure}[ht!]
    \begin{tabular}{|c| c|c|c|c|}
      \hline
SS SLEEPING & round complexity & awake complexity   & awake frequency & Scheduler \\
    transformer   &    & before convergence  &  after convergence & requirement \\ 
     \hline \hline
Algorithm~\ref{alg-sleeping-simple} & $T.SP$ & $T$ &  $1/SP$ & weakly \\
& & & & fair distributed
 \\ 
\hline
Algorithm~\ref{alg-sleeping-fast}  & $3T+(3D+2).SP$ & $3T+(3D+2).SP$ &  $1/SP$ & synchronous\\
\hline
\end{tabular}
\caption{Our Self-Stabilizing SLEEPING tranformers.}
\label{fig:transformers}
\end{figure}
Note that, in~\cite{DIJM24}, there is a transformer that turns any silent synchronous algorithm into an asynchronous self-stabilizing one with the same complexity.
By applying either of the two transformers in Figure~\ref{fig:transformers}, we can turn any silent synchronous algorithm into a self-stabilizing SLEEPING algorithm.


\section{Model}
\label{sec:Model}

\subsection{Self-Stabilizing Algorithms in Static Graphs}
\paragraph*{Communication and computation models}
We consider a distributed system modeled by a non-oriented connected graph $G$.
Two nodes are neighbors in $G (V, E)$ if and only if they can communicate with each other. The number of nodes is denoted by $n$, $D$ denotes the diameter of $G$, and $\Delta$ is the maximal degree of $G$.
The set of neighbors of a node $u\in V$ is denoted by $N(u)$; the closed neighborhood of $u$ (i.e., $N(u) \cup \{u\}$ is denoted by $N[u]$.
In that paper, we study \emph{identified} networks, meaning that every node $u$ has a unique identifier, denoted by $\operatorname{id}_u$.

Every node has a finite set of \textit{variables}.
Some variables, such as identifier, are \textit{constant} inputs provided by the system.
Each node can read its own variables and variables owned by the neighboring nodes, but can write only to its own (non-constant) variables.
The \textit{state} of a node $u$ designates the current value of all its local variables.
An \textit{action} by a node $u$ is the updating of some non-constant variables of $u$ according to its algorithm, its state, and the state of its neighbors.

\paragraph*{Configuration, Execution, Asynchrony.}
A {\em configuration} of an algorithm ${\mathcal A}$ for $V$ is a vector $(s_1, \ldots,s_n)$, where $s_1$ to $s_n$ represent the states of all nodes of $G$.
A node is said to be \textit{enabled} in a configuration $\gamma$ if and only if it can perform an action in $\gamma$. 
If no node is enabled in a configuration $\gamma$, then $\gamma$ is a \textit{terminal} configuration.

We model the asynchrony of the system by a schedule: a function that takes as an input the algorithm and a sequence of configurations $\gamma_0,\dots,\gamma_t$, and returns the empty set if $\gamma_t$ is terminal, and a non-empty subset of the enabled nodes in $\gamma_t$ otherwise.
A scheduler is a class of schedules that share common properties, such as the level of fairness they guarantee regarding when nodes are scheduled.
A taxonomy of schedulers has been proposed in~\cite{DT11}.
In this paper, we only consider two schedulers.
The \textit{synchronous scheduler} ensures that every enabled node updates its variables at each computation step.
The \textit{weakly fair distributed scheduler} ensures that a node cannot be enabled in every configuration of an execution without eventually being scheduled; there exists a computation step in which the node is either scheduled or neutralized.

Given a schedule $\mathcal D$ and an algorithm ${\mathcal A}$, an \emph{execution} of $\mathcal A$ under $\mathcal D$ is a finite or infinite sequence of configuration $\gamma_0,\gamma_1,\dots$ such that $\forall i > 0, \gamma_i$ is the configuration obtained from $\gamma_{i-1}$ after the execution of one simultaneous action by all nodes of $\mathcal D(\mathcal A, \gamma_0,\gamma_1,\dots,\gamma_{i-1})$, the scheduled nodes.
An execution is \textit{maximal} if it is infinite or the last configuration is terminal.
A \emph{computing step} of an execution $\gamma_0,\gamma_1,\dots$ is any tuple $\gamma_i\gamma_{i+1}$ for $i\geq 0$.
If a node $u$ has a variable $\operatorname{var}$, we note $\operatorname{var}_u^r$ its value in configuration $\gamma_r$, or just $\operatorname{var}_u$ if it is unambiguous.

\paragraph*{Self-stabilization.}
Let $\mathcal A$ be a distributed algorithm, $\mathcal{S}$ be a specification (i.e., a predicate over graphs and configuration sequences), and $\mathcal D$ be a scheduler.
\begin{definition}[Self-stabilization]
  \label{def:ss-static}
  Algorithm $\mathcal A$ is self-stabilizing for $\mathcal{S}$ under $\mathcal D$ 
  if for any graph $G$, there exists a subset of configurations $\mathcal L(G)$, 
  called legitimate configurations, such that
  \begin{enumerate}
    \item
      for every configuration $\gamma$ of $\mathcal A$ in $G$, any maximal execution of $\mathcal A$ in $G$ under $\mathcal D$ starting from $\gamma$ contains a legitimate configuration $\gamma' \in \mathcal L(G)$ (Convergence), and
    \item
      for every legitimate configuration $\gamma \in \mathcal L(G)$, any execution $\epsilon$ of $\mathcal A$ in $G$ under $\mathcal D$ starting from $\gamma$ satisfies $\mathcal{S}(G, \epsilon)$ (Correctness).
\end{enumerate}
\end{definition}
If $\mathcal L(G)$ is the set of terminal configurations
 then Algorithm $\mathcal A$ is a \textit{silent} algorithm.

\paragraph*{Convergence time.}
A node $u$ is \emph{neutralized\/} in the computing step $\gamma \rightarrow \gamma'$ if $u$ is enabled in $\gamma$ and not enabled in $\gamma'$, but does not execute any action between these two configurations.
Neutralization occurs when some neighbors of $u$ change their state between $\gamma$ and $\gamma'$, and this change brings $u$ into a situation where $\mathcal{A}$ applied to $u$ does not change its state.
To evaluate the stabilization time, the classical time unit used in the literature is the \emph{Asynchronous round}~\cite{CDDPV16}.
The first asynchronous round of an execution $\epsilon$, noted $\epsilon^{\prime}$, 
is the minimal prefix of $\epsilon$ in which every node that is enabled in the initial configuration either executes an action or becomes neutralized.
The second asynchronous round of $\epsilon$ is the first round of the suffix of $\epsilon$ starting from the last configuration of $\epsilon^{\prime}$, and so on. 
In the case of the synchronous scheduler, each computation step corresponds to a new round.

\subsection{SLEEPING Self-Stabilizing Algorithm}

In the SLEEPING Self-Stabilizing model, nodes still have access to their neighbors' current states, even when sleeping. It differs from the SLEEPING LOCAL model, in which both nodes must be awake to share their current state. In that model, knowledge of a shared round 0, in which all nodes can be awake and aware of their neighbors, allows for the development of elaborate algorithms that cannot be easily adapted to self-stabilization. Self-stabilizing algorithms, in which both nodes must be awake simultaneously to share their state (or even be aware of each other), have been extensively studied in the wireless sensor network model.

To be able to have a notion of awake rounds in our model, we add the following elements:

\paragraph*{Internal clock.}
Each node has a local variable $cl$, called the internal clock, that cannot be read by its neighbors. It is a natural number representing the number of rounds the node waits before activating itself. In each round, the node will first decrement $cl$ by 1. Then, if $cl>0$, the node does not update any of its variables and is not enabled. Otherwise, it is awake. So it executes an action if enabled, based on the current state of its neighborhood. 
It then updates $cl$ to the number of rounds the node must sleep before its next activation 
($cl$ is set to 1 if the node gets to be awake in the next round).

Note that to describe a configuration, we need to know the value of each internal clock. 
An algorithm is considered \textit{silent} if, after a legitimate configuration has been reached, every node variable except the clock remains unchanged.

\paragraph*{Synchrony.}
 In this model, we consider synchronous steps. In each step, each node updates itself, first by decreasing its clock and then by applying an action if that clock hits 0. 
In particular, it means that each step is an actual (asynchronous) round in the Self-Stabilizing model. 
In the rest of the paper, we will directly talk about rounds for each synchronous step the system goes through. 
Configuration $\gamma_r$ corresponds to the states of the system after round $r$. 
We say that a node is awake or activated in round $r$ if, in $\gamma_{r-1}$, its clock has value $1$, and will execute the algorithm.

\paragraph*{Complexities.}
In this model, for a given algorithm $\mathcal{A}$, we are considering three complexities:
\begin{itemize}
    \item Round complexity: the maximal number of (synchronous) rounds needed, from any configuration, to reach a legitimate configuration. 
    \item Awake complexity before convergence: From a given execution $\epsilon= \gamma_0, \gamma_1 \cdots$, given a node $v$, $a^\epsilon_v$ corresponds to how many times $v$ is awake before a legitimate configuration $\gamma _l$  is reached for the first time during $\epsilon$.
    The awake complexity before convergence of $\mathcal{A}$ is the maximal value among all executions and all nodes of $a^\epsilon_v$.
    \item Frequency after convergence: let $\epsilon$ be an execution starting from a legitimate configuration. During $\epsilon$, the frequency of activation of a node $v$ is the average duration between two consecutive rounds where $v$ is awake. 
     The frequency after convergence of $\mathcal{A}$ is the maximal value among all executions starting from a legitimate configuration and all nodes of their frequencies of activation.
\end{itemize}


\section{$(\Delta+1)$-Coloring}\label{sec:col}

In this section, we introduce a SLEEPING Self-Stabilizing algorithm that computes a $(\Delta+1)$-coloring in $O(\Delta)$ awake rounds.

\subsection{Algorithm}

The idea of Algorithm~\ref{alg:coloring} is, for each node, to find a round in which it will be the only one to change its color. 
To that end, every node $u$ picks a unique prime number $\operatorname{p}(\operatorname{id_u})$ large enough, based on its identifier $\operatorname{id_u}$ and $\Delta$. 
As long as a node $u$ is in conflict with one of its neighbors, it will wake up every $\operatorname{p}(\operatorname{id_u})$ rounds to select the smallest available color. 
As each node is using a different prime number, we prove that, during $\Delta +1$ consecutive attempts of $u$ to set up a proper color, there are not 2 rounds in which a neighbor $v$ of $u$ changes also its color. 
Hence, after at most $\Delta+1$ attempts, $u$ will change its color while all of its neighbors are sleeping. 
From then on, $u$ keeps its color forever. 
From that point on, $u$ will wake up every $SP$ rounds, for some sufficiently large shared constant $SP$.

\begin{algorithm}
\caption{Energy-efficient self-stabilizing $(\Delta+1)$-coloring algorithm}\label{alg:coloring}
\textbf{Input:}\\
\quad $\operatorname{id}_u$: identifier of $u$ \\
\quad $M$: a strict upper bound on the largest identifier\\
\quad $\Delta$: the maximum degree of the graph\\
\quad $SP$: an integer being the awake frequency, strictly greater than $\Delta \operatorname{p}(M)$\\

\textbf{Macro:}\\
\quad $\operatorname{p}(x)$: the $x$-th prime number strictly greater than $\Delta + 1$\\

\textbf{Shared variables:}\\
\quad $\operatorname{color}_u$: a color in $[ 0, \Delta]$\\

\textbf{Predicates:}\\
\quad $\operatorname{conflict}(u) \equiv \exists v \in N(u), \operatorname{color}_u =  \operatorname{color}_v$\\

\SetKwBlock{UponWaking}{Upon waking:}{}
\UponWaking{
\nl	\If{$\operatorname{conflict}(u)$}{
\nl\label{acol:l2}        $\operatorname{color}_u:=\min [ 0, \Delta ] \setminus \{ \operatorname{color}_v | v \in N(u) \}$\\
\nl\label{acol:l3}		Sleep for $\operatorname{p}(\operatorname{id}_u)$ rounds.
    }\nl \Else {
\nl		Sleep for $SP$ rounds.
	}
}
\end{algorithm}

\subsection{Analysis}

To prove that Algorithm~\ref{alg:coloring} is correct and its complexity, we use the following key lemma that ensures that if a node has a conflict, it will be solved:

\begin{lemma}\label{lem:color-conflict}
Let $\gamma_0$ be a configuration in which some node $u$ is such that predicate $\operatorname{conflict(u)}$ is satisfied. In the rounds numbered from $1$ to $ r = SP + \Delta\operatorname{p}(\operatorname{id}_u)$, node $u$ is awake at least once with predicate $\operatorname{conflict(u)}$ unverified, or being the only node of its neighborhood to be awake.
\end{lemma}
\begin{proof} 
Suppose that every time $u$ is awake between rounds $1$ and $r$, $\operatorname{conflict}(u)$ is True, and at least one of its neighbors is awake.

First, observe that $u$ will be awake at least once before or at round $SP$, as it sleeps each time for at most $SP$ rounds. 
Since $\operatorname{conflict}(u)$ is at True each time it is awake, it will each time go to sleep for  $\operatorname{p}(\operatorname{id}_u)$ rounds. 
Hence, it will be awake at least $\Delta+1$ times. 

Since $u$ has at most $\Delta$ neighbors, one of them, say $v$, is awake at least twice at the same round as $u$ among the $\Delta + 1$ rounds in which $u$ is awake. Let's call $r_1$ the first round in which both $u$ and $v$ are awake at the same time, and $r_2$ the next. Note that $r_1 \leq r_2 \leq r_1 + \Delta \operatorname{p}(\operatorname{id}_u)$.

Since $\Delta \operatorname{p}(\operatorname{id}_u) < T$, node $v$ also does not go sleeping for $T$ rounds between $r_1$ and $r_2$, so it is awake every $\operatorname{p}(\operatorname{id}_v)$ rounds.
Hence, $r_1 = r_2 + \operatorname{LCM} ( \operatorname{p}(\operatorname{id}_u),  \operatorname{p}(\operatorname{id}_v))$, where $\operatorname{LCM}(x,y)$ is the \textit{Lowest Common Multiple} of $x$ and $y$.

However, $\operatorname{p}(\operatorname{id}_u)$ and $\operatorname{p}(\operatorname{id}_u)$ are two different prime numbers, so $\operatorname{LCM} ( \operatorname{p}(\operatorname{id}_u), \operatorname{p}(\operatorname{id}_v)) =  \operatorname{p}(\operatorname{id}_u) \operatorname{p}(\operatorname{id}_v)$. Since $\operatorname{p}(\operatorname{id}_v) > \Delta + 1$, this leads to a contradiction with $r_2 \leq r_1 +  \Delta \operatorname{p}(\operatorname{id}_u) $.
\end{proof}

\begin{lemma}\label{lem:no-more-conflict}
Let $\gamma_0$ be a configuration and $u$ such that $\operatorname{conflict}(u) = \operatorname{False}$. We get that $\operatorname{conflict}(u) = \operatorname{False}$ on any configuration $\gamma_r$ with $r\in\mathbb{N}$.
\end{lemma}
\begin{proof} 
Assume by contradiction that there exists some $r$ such that $\operatorname{conflict}(u) = \operatorname{True}$ in $\gamma_r$. Let us choose the smallest $r$. It means that in round $r$, a node $v\in N[u]$ has changed its color, and we have $\operatorname{color}_u^r=\operatorname{color}_v^r$.

We cannot have $v=u$, as $u$ was not in conflict in $\gamma_{r-1}$, by the choice of $r$. In round $r$, $v$ has applied line~\ref{acol:l2} of the algorithm. It means that $v$ took a color different from $\operatorname{color}_u^{r-1}$. As $u$ did not change its color, we get that $\operatorname{color}_u^r=\operatorname{color}_u^{r-1}$, which leads to a contradiction.
\end{proof}

The following lemma allows us to get the complexity of $\p(M)$:

\begin{lemma}\label{lem:mlogm}
For any $k\in\mathbb{N}$ large enough, the $k$th prime number is $O(k\log k)$.
\end{lemma}
\begin{proof}
Let $\pi(x)$ be the number of prime numbers smaller than or equal to $x$. 
By the Prime Number Theorem~\cite{hadamard1896distribution}, we have:
$$\pi(x)=\frac{x}{\log x}+o\left(\frac{x}{\log x}\right)$$

In particular, $\pi(2k\log k)=\frac{2k\log k}{\log k}+o\left(\frac{2k\log k}{\log k}\right)=2k+o(k)$. It implies that for $k$ large enough, the $k$th prime number is at most $2k\log k$.
\end{proof}

\begin{theorem}
Algorithm~\ref{alg:coloring} computes a $(\Delta+1)$-coloring in $SP+\Delta\operatorname{p}(M)= SP+ O(\Delta M \log M) $ rounds. Each node in those rounds is awake at most $\Delta+2=\Delta+O(1)$ times. The awake frequency after convergence is $1/SP$.
\end{theorem}
\begin{proof}
The computation of a $(\Delta+1)$-coloring comes from the fact that each node will reach a configuration in which it has no conflict (Lemma~\ref{lem:color-conflict}), from which it will never be in conflict again (Lemma~\ref{lem:no-more-conflict}). Lemma~\ref{lem:color-conflict} tells us that a node will apply lines~\ref{acol:l3} at most $\Delta+1$ times each and otherwise sleep each time for $T$ rounds, which gives the awake complexity and awake frequency after convergence. 

By choosing $SP=\Delta\operatorname{p}(M)$, we get a convergence time of $O(\Delta \operatorname{p}(M))$ and an awake frequency after convergence of $1/(\Delta \operatorname{p}(M))$.
Let $M>  \Delta +1$ be an upper bound on the larger identifier of a node in the network.
Lemma~\ref{lem:mlogm} allows us to conclude that $\p(M)=O(M\log M)$, giving us the complexities of Figure~\ref{fig:algs}.
\end{proof}

\section{Maximal Independent Set}\label{sec:mis}

In this section, we introduce a SLEEPING Self-Stabilizing algorithm that computes a Maximal Independent Set in $O(\Delta)$ awake rounds.

\subsection{Algorithm}

To build a maximal independent set, nodes must be either in or out of it. In our algorithm, when both a node $u$ and one of its neighbors are in the independent set, $u$ can freely leave the set. 
However, in the case where neither $u$ nor any of its neighbors is in the independent set, we must ensure that $u$ joins the set only if none of its neighbors will do so at the same time, as otherwise we could create a loop where a node will alternately join and leave the set. 
To that end, the action to join the independent set takes two rounds. 
First, the node sets the variable $\operatorname{justWokeUp}$ to True. 
If during the next round, none of its neighbors has this variable set to True, it joins the independent set; otherwise, $u$ will sleep for $\operatorname{p}(\operatorname{id}_u-1)$ before considering joining the independent set again.
This technique ensures that nodes join the independent set via a \emph{mutual exclusion} process: no pair of neighbors can join simultaneously.

The idea of Algorithm~\ref{alg:mis} is similar to the $(\Delta+1)$-coloring algorithm: we use prime numbers to ensure that during a round, a node is the only one in its neighborhood to set the variable $\operatorname{justWokeUp}$ to True.

\begin{algorithm}
\caption{Energy-efficient self-stabilizing maximal independent set algorithm}\label{alg:mis}
\textbf{Input:}\\
\quad $\operatorname{id}_u$: identifier of $u$ \\
\quad $M$: a strict upper bound on the largest identifier\\
\quad $\Delta$: the maximum degree of the graph\\
\quad $SP$: an integer being the awake frequency, strictly greater than $\Delta \operatorname{p}(M)$\\

\textbf{Macro:}\\
\quad $\operatorname{p}(x)$: the $x$-th prime number strictly greater than $\Delta + 1$\\

\textbf{Shared variables:}\\
\quad $\operatorname{justWokeUp}_u$: True if $u$ just woke up on this round, else False\\
\quad $\operatorname{inMIS}_u$: True if $u$ is in the maximal independent set, else False\\

\textbf{Predicates:}\\
\quad $\operatorname{mustLeave}(u) \equiv \operatorname{inMIS}_u \land \exists v \in N(u), \operatorname{inMIS}_v$\\
\quad $\operatorname{canJoin}(u) \equiv \lnot \operatorname{inMIS}_u \land \forall v \in N(u), \lnot \operatorname{inMIS}_v$\\
\quad $\operatorname{awakeAlone}(u) \equiv \operatorname{justWokeUp}_u \land \forall v \in N(u), \lnot \operatorname{justWokeUp}_v$\\

\SetKwBlock{Begin}{Upon waking:}{}
\Begin{
\nl	\If{$\operatorname{canJoin}(u)$}{	
\nl		\If{$\operatorname{justWokeUp}_u$}{
\nl        	$\operatorname{justWokeUp}_u:=\operatorname{False}$\\
\nl			\If{$\operatorname{awakeAlone}(u)$}{
\nl    			$\operatorname{inMIS}_u:=\operatorname{True}$
    		}
\nl\label{aMIS:l6}    		Sleep for $ \operatorname{p}(\operatorname{id}_u) - 1$ rounds.
    	}\nl \Else {
\nl\label{aMIS:l9} $\operatorname{justWokeUp}_u:=\operatorname{True}$\\
\nl    		Sleep for 1 round.
    	}
    }
\nl     \Else {
\nl        \If{$\operatorname{justWokeUp}(u)$}{	
\nl            $\operatorname{justWokeUp}_u:=\operatorname{False}$
        }
\nl\label{aMIS:l14}        \If{$\operatorname{mustLeave}(u)$}{	
\nl           $\operatorname{inMIS}_u:=\operatorname{False}$
        }
\nl\label{aMIS:l15}    	Sleep for $SP$ rounds.
    }
}
\end{algorithm}

\subsection{Analysis}
If a node $u$ wakes up with the variable $\operatorname{justWokeUp}$ at True, we say that \textbf{$u$ is actively awake}.\\

\textbf{Observation}: Let $\gamma$ be a configuration, and let $u$ and $v$ be two neighboring nodes. We cannot have both $u$ and $v$ setting their respective variable $\operatorname{inMIS}$ to True during the round from $\gamma$. 
Indeed, to do so, we would have $\operatorname{awakeAlone}(u)=\operatorname{awakeAlone}(v)=\operatorname{True}$, which is a contradiction of the definition of the predicate $\operatorname{awakeAlone}$.

\begin{lemma}\label{lem:activelyawake}
Let $\gamma_0$ be a configuration in which some node $u$ verifies the predicate $\operatorname{canJoin(u)}$. 
In the rounds numbered from $1$ to $ r = SP + \Delta \operatorname{p}(\operatorname{id}_u)$, node $u$ is awake at least once with predicate $\operatorname{canJoin(u)}$ unverified, or being the only node of its neighborhood to be actively awake.
\end{lemma}
\begin{proof} 
Suppose that every time $u$ is awake between rounds $1$ and $r$, $\operatorname{canJoin}(u)$ is at True, and $\operatorname{awakeAlone}(u)$ is at False in the rounds in which $u$ is actively awake (i.e., in each configuration where $\operatorname{justWokeUp}_u=\operatorname{True}$, $u$ has a neighbor $v$ such that $\operatorname{justWokeUp}_v=\operatorname{True}$).

First, observe that $u$ will be awake at least once before round $SP$. Since $\operatorname{canJoin}(u)$ is at True each time it is awake, it never goes sleeping for $SP$ rounds and will be actively awake every $\operatorname{p}(\operatorname{id}_u)$ rounds (by construction of the algorithm). By the choice of $r$, this will happen at least $\Delta+1$ times. 

Since $u$ has at most $\Delta$ neighbors, one of them, say $v$, is actively awake at least twice at the same round as $u$ among the $\Delta + 1$ rounds in which $u$ is actively awake. Let's call $r_1$ the first round in which both $u$ and $v$ are awake at the same time, and $r_2$ the next. Note that $r_1 \leq r_2 \leq r_1 + \Delta  \operatorname{p}(\operatorname{id}_u) $.

Since $\Delta \operatorname{p}(\operatorname{id}_u) < SP$, node $v$ also does not go sleeping for $SP$ rounds between $r_1$ and $r_2$, so it is actively awake up every $\operatorname{p}(\operatorname{id}_v)$ rounds.
Hence, $r_1 = r_2 + \operatorname{LCM} ( \operatorname{p}(\operatorname{id}_u),  \operatorname{p}(\operatorname{id}_v))$.

However, $\operatorname{p}(\operatorname{id}_u)$ and $\operatorname{p}(\operatorname{id}_u)$ are two different prime numbers, so $\operatorname{LCM} ( \operatorname{p}(\operatorname{id}_u), \operatorname{p}(\operatorname{id}_v)) =  \operatorname{p}(\operatorname{id}_u) \operatorname{p}(\operatorname{id}_v)$. Since $\operatorname{p}(\operatorname{id}_v) > \Delta + 1$, this leads to a contradiction with $r_2 \leq r_1 +  \Delta \operatorname{p}(\operatorname{id}_u) $.
\end{proof}

\begin{lemma}\label{lem:mustLeaveFalse}
Let $u$ be a node.
Let  $\gamma_0, \gamma_1, \cdots$ be an execution starting from a configuration where $\operatorname{inMIS}_u = \operatorname{False}$. 
 We have $\operatorname{mustLeave}(u) = \operatorname{False}$ on any configuration $\gamma_r$ with $r\geq 0$.
\end{lemma}
\begin{proof} 
Assume that in $\gamma_r$ we have $\operatorname{mustLeave}(u) = \operatorname{True}$. Then, in $\gamma_r$, $u$ has a neighbor $v$ such that $\operatorname{inMIS}_v = \operatorname{True}$.

By definition of $\operatorname{mustLeave}(u)$, the variable $\operatorname{inMIS}_u= \operatorname{True}$ in $\gamma_r$. 
Let's call $r_1$ the last round before $r$ in which $\operatorname{inMIS}_u = \operatorname{False}$. Note that $r_1\geq0$.

On round $r_1$, $u$ sets the value of $\operatorname{inMIS}_u$ to True, so $u$ is the only node in its neighborhood to be actively awake, and every neighbor of $u$ has its variable $\operatorname{inMIS}$ set to False. 
This is in particular the case for $v$. 
Hence, there is a round $r_2$ between $r_1$ and $r$ where $v$ changes the value $\operatorname{inMIS}_v$ from False to True. 
In the configuration before that round, we have $\operatorname{canJoin}(v) = \operatorname{True}$.
However, in that same configuration, we have $\operatorname{inMIS}_u = \operatorname{True}$, so $\operatorname{canJoin}(v) = \operatorname{False}$. This leads to a contradiction.
\end{proof}

\begin{lemma}\label{lem:misremainsmis} 
From any configuration $\gamma_0$, for any $r\ge SP$, the set $\{u | \operatorname{inMIS}_u=\operatorname{True} \}$ in $\gamma_r$ is independent.
\end{lemma}
\begin{proof}
Suppose there exist two adjacent nodes $u$ and $v$ such that $\operatorname{inMIS}_u = \operatorname{inMIS}_v = \operatorname{True}$ on a round $r \ge SP$. Then $\operatorname{mustLeave}(u) = \operatorname{mustLeave}(v) = \operatorname{True}$ on round $r$.

By Lemma~\ref{lem:mustLeaveFalse}, the value of $\operatorname{inMIS}$ is True for both nodes on every configuration before round $r$. 
However, both $u$ and $v$ were awake at least once before round $r$, as a node can sleep for at most $SP$ rounds. Out of generality, assume that $u$ was awake first in some round $r_1<r$. In that round, it must have applied the action from line~\ref{aMIS:l14} and set its variable $\operatorname{inMIS}_u$ to False. This contradicts the fact that $\operatorname{mustLeave}(u)$ remains True up to round $r$.

This implies that from round $SP$ on, there is no pair of adjacent nodes with both of them having their variable $\operatorname{inMIS}$ at True, which concludes the proof.
\end{proof}

\begin{lemma}\label{lem:IS-increases}
Let $\gamma_0$ be a round in which $\{u | \operatorname{inMIS}_u=\operatorname{True} \}$ is an independent set, and let $u$ be a node such that $\operatorname{inMIS}_u=\operatorname{True}$ in that configuration. For any $r >0$, we have $\operatorname{inMIS}_u^r=\operatorname{True}$.
\end{lemma} 
\begin{proof}
 Suppose there exists $r>0$ such that $\operatorname{inMIS}_u = \operatorname{False}$.
 
Without loss of generality, let's say $r$ is the smallest possible choice. Then $\operatorname{inMIS}_u = \operatorname{True}$ on any round before, and in particular, node $u$ changes the value of $\operatorname{inMIS}_u$ to False on round $r$. It implies that $u$ has a neighbor $v$ with $\operatorname{inMIS}_v = \operatorname{True}$ on this round.

By Lemma~\ref{lem:misremainsmis}, $\{u | \operatorname{inMIS}_u=\operatorname{True} \}$ is an independent set in configuration $\gamma_{r-1}$, but we have $\operatorname{inMIS}_u = \operatorname{inMIS}_v = \operatorname{True}$. This is a contradiction.
\end{proof}

\begin{lemma}\label{lem:MIS}
Let $\gamma_0$ be a round in which $\{u | \operatorname{inMIS}_u=\operatorname{True} \}$ is an independent set. For any  $r > SP + \Delta  \operatorname{p}(M)$, the set $\{u | \operatorname{inMIS}_u=\operatorname{True} \}$ in $\gamma_r$ is maximal independent.
\end{lemma}
\begin{proof}
Let $r > SP+\Delta \operatorname{p}(M)$. By Lemma~\ref{lem:misremainsmis}, the set $\{u | \operatorname{inMIS}_u=\operatorname{True} \}$ is independent in $\gamma_i$ with $i\leq r$.
Suppose this set is not maximal. 
It implies that there exists a node $u$ such that $\operatorname{canJoin}(u) = \operatorname{True}$ in this configuration.

Then, each node $v$ in the closed neighborhood of $u$ verifies $\operatorname{inMIS}_v^r = \operatorname{False}$. 
Lemma~\ref{lem:IS-increases} implies directly that $\operatorname{inMIS}_v = \operatorname{False}$ in any configuration from $\gamma_{0}$ to $\gamma_r$. 
Hence, in $\gamma_{0}$, $\operatorname{canJoin(u)}$ is verified.

By Lemma~\ref{lem:activelyawake}, there exists some round $r_1\le r$ in which either $\operatorname{canJoin(u)}$ is unverified (case 1) or $u$ was the only actively awake node (case 2).

In case 1, it implies that one node $v\in N[u]$ (which can be $u$) has  $\operatorname{inMIS}_v^{r_1}
= \operatorname{True}$. 
Lemma~\ref{lem:IS-increases} ensures that $\operatorname{inMIS}_v$ remains True forever. This implies that currently, in $\gamma_{r}$, $\operatorname{canJoin}(u) = \operatorname{False}$.

In case 2, in round $r_1$, $u$ sets $\operatorname{inMIS}_u$ to True, and by the same reasoning, this variable will never change its value and $\operatorname{canJoin}(u)^r = \operatorname{False}$.
\end{proof}

\begin{theorem}
Algorithm~\ref{alg:mis} computes a Maximal Independent Set in $2SP+\Delta\operatorname{p}(M)$ rounds. Each node in those rounds is awake at most $2\Delta+4=O(\Delta)$ times. The awake frequency after convergence is $1/SP$.
\end{theorem}
\begin{proof}
The computation of a Maximal Independent Set comes from Lemmas~\ref{lem:misremainsmis} and \ref{lem:MIS}. Lemma~\ref{lem:activelyawake} tells us that a node will apply lines~\ref{aMIS:l6} and \ref{aMIS:l9} at most $\Delta+1$ times each and otherwise sleep each time for $SP$ rounds, which gives the awake complexity and awake frequency after convergence. 

By choosing $SP=\Delta\operatorname{p}(M)$, we get a convergence time of $O(\Delta \operatorname{p}(M))$ and an awake frequency after convergence of $1/(\Delta \operatorname{p}(M))$, matching the complexities of Figure~\ref{fig:algs}.
\end{proof}

\section{Maximal Matching}\label{sec:mm}

This section presents an energy-efficient self-stabilizing algorithm for computing a maximal matching.

\subsection{Algorithm}

\begin{algorithm}
\caption{Energy-efficient self-stabilizing maximal matching algorithm}\label{alg:mm}
\DontPrintSemicolon
\textbf{Input:}\\
\quad $\operatorname{id}_u$: identifier of $u$ \\
\quad $M$: a strict upper bound on the largest identifier\\
\quad $\Delta$: the maximum degree of the graph\\
\;
\textbf{Macro:}\\
\quad $B(M)$: the smallest integer such that the interval $ [ B(M), 2B(M) ]$ contains $M$ prime number\\
\quad $\p_M(x)$: the $x$-th prime number in $ [ B(M), 2B(M) ]$\\
\;
\textbf{Shared variables:}\\
\quad $\operatorname{justWokeUp}_u$: True if $u$ just woke up on this round, else False\\
\quad $\operatorname{status}_u \in \{ \operatorname{inMIS}, \operatorname{outMIS}, \operatorname{matched} \}$: the current status of the node\\
\quad $\operatorname{partner}_u \in N(u)$: the partner of $u$, only relevant if $u$ is matched\\
\quad $\operatorname{seducing}_u \in N(u)$: the node seduced by $u$, only relevant if $u$ is inside of the MIS\\
\;
\textbf{Predicates:}\\
\quad $\operatorname{awakeAlone}(u) \equiv \operatorname{justWokeUp}_u \land \forall v \in N(u), \lnot \operatorname{justWokeUp}_v$\\
\quad $\operatorname{canJoinMIS}(u) \equiv \operatorname{status}_u = \operatorname{outMIS} \land \forall v \in N(v), \operatorname{status}_v \neq \operatorname{inMIS}$\\
\quad $\operatorname{mustLeaveMIS}(u) \equiv \operatorname{status}_u = \operatorname{inMIS} \land \exists v \in N(v), \operatorname{status}_v = \operatorname{inMIS}$\\
\quad $\operatorname{married}(u) \equiv \operatorname{status}_u = \operatorname{status}_{\operatorname{partner}_u} = \operatorname{matched} \land \operatorname{partner}_{\operatorname{partner}_u} = u$\\
\quad $\operatorname{waiting}(u) \equiv \operatorname{status}_u = \operatorname{matched} \land \operatorname{status}_{\operatorname{partner}_u} = \operatorname{inMIS} \land \operatorname{seducing}_{\operatorname{partner}_u} = u$\\
\quad $\operatorname{badMatching}(u) \equiv \operatorname{status}_u = \operatorname{matched} \land \lnot \operatorname{married}(u) \land \lnot \operatorname{waiting}(u)$\\
\quad $\operatorname{gotMatch}(u) \equiv \operatorname{status}_u = \operatorname{inMIS} \land \operatorname{status}_{\operatorname{seducing}_u} = \operatorname{matched}\land \operatorname{partner}_{\operatorname{seducing}_u} = u$\\
\quad $\operatorname{gotSeduced}(u) \equiv \operatorname{status}_u = \operatorname{outMIS} \land \exists v \in N(u), \operatorname{status}_v = \operatorname{inMIS} \land \operatorname{seducing}_v = u$\\
\quad $\operatorname{canSeduce}(u) \equiv \operatorname{status}_u = \operatorname{inMIS} \land \operatorname{status}_{\operatorname{seducing}_u} \neq
\operatorname{outMIS}$\\ \hspace{5.43cm}$\land \exists v \in N(u), \operatorname{status}_v = \operatorname{outMIS}$\\

\;
\SetKwBlock{UponWaking}{Upon waking:}{}
\UponWaking{
\nl	\If{$\operatorname{canJoinMIS}(u)$}{
\nl		\If{$\operatorname{justWokeUp}_u$}{
\nl\label{amm:l3}			\lIf{$\operatorname{awakeAlone}(u)$}{
    			$\operatorname{status}_u:= \operatorname{inMIS}$
    		}
\nl    		$\operatorname{justWokeUp}_u:=\operatorname{False}$;
    		Sleep for $ \p_M(\operatorname{id}_u) - 1$ rounds.
    	}\nl\label{amm:l5} \lElse {
        	$\operatorname{justWokeUp}_u:=\operatorname{True}$;
    		Sleep for 1 round.
    	}
	}\nl \Else {
\nl		\lIf {$\operatorname{justWokeUp}_u$} {
        	$\operatorname{justWokeUp}_u:=\operatorname{False}$
        }
\nl		\lIf {$\operatorname{mustLeaveMIS}(u) \lor \operatorname{badMatching}(u)$} {
			$\operatorname{status}_u := \operatorname{outMIS}$
		}\nl \Else {
\nl			\If{$\operatorname{gotMatch}(u)$} {
\nl				$\operatorname{partner}_u := \operatorname{seducing}_u$ ;
				$\operatorname{status}_u := \operatorname{matched}$
    		}
\nl    		\If {$\operatorname{gotSeduced}(u)$} {
\nl\label{amm:l13}    			$\operatorname{partner}_u := \text{argmin}_{v\in N(u)} \{ id_v, \operatorname{status}_v = \operatorname{inMIS} \land \operatorname{seducing}_v = u \}$\\
\nl\label{amm:l14}    			$\operatorname{status}_u := \operatorname{matched}$
    		}
\nl    		\lIf {$\operatorname{canSeduce}(u)$} {
\nl    			$\operatorname{seducing}_u := \text{argmin}_{v \in N(u)} \{ id_v, \operatorname{status}_v = \operatorname{outMIS} \}$
			}
    	}
\nl    	Sleep for $ \p_M(\operatorname{id}_u)$ rounds.
    }
}
\end{algorithm}

The idea of Algorithm~\ref{alg:mm} is inspired by the constructions from~\cite{cohen2023state}. We use Algorithm~\ref{alg:mis} as a subroutine to compute a maximal independent set $IS$, which will be used to perform some mutual exclusion. Each node $u$ in $IS$ attempts to match with an unmatched node $v$ via a seducing step. Either $v$ accepts, and both nodes are matched and ignored by the independent set subroutine, or $v$ accepts another seducing node, and $u$ tries to match itself with another node. If all of its attempts are rejected, it implies that all of its neighbors are matched. A node in $IS$ leaves the independent set once it is matched. After that, the subroutine from Algorithm~\ref{alg:mis} tries to add new nodes in $IS$. That way, an unmatched node will eventually join $IS$ and try to seduce its unmatched neighbors.

To ensure that the seducing process is not too long, we need to ensure that a node from $IS$ does not wait too long for an answer. To that end, we add a restriction to the prime numbers we use. They must belong to a set $[B, 2B]$, for some $B$ large enough (this set must contain at least $M$ prime numbers; the proof of existence of such $B$ can be found in Lemma~\ref{lem:Bexists}). With this restriction, we know that for any pair of neighbors $u$ and $v$, if $u$ (resp. $v$) wakes up every $\p_M(\operatorname{id}_u)$ (resp. $\p_M(\operatorname{id}_v)$) rounds, $v$ will be awake at least once between 3 consecutive awake rounds of $u$.

In this algorithm, each node $u$ wakes up every $\p_M(\operatorname{id}_u)$ rounds. In each of those rounds, we say that the node is \emph{commonly awake}. A node can also be awake in one more round every $\p_M(\operatorname{id}_u)$ rounds, for the $\operatorname{justWokeUp}$ part of the maximal independent set subroutine. This is necessary as computing nodes need feedback from their neighbor. This system ensures some pseudo-synchrony. Contrarily to Algorithms~\ref{alg:coloring} and~\ref{alg:mis}, the frequency of waking up of a node does not change after having finished its computation. This is due to the fact that in maximal matching, an edge can be in the independent set if \textit{both} endpoints are happy about it.

\subsection{Analysis}\label{sec:proof-mm}

\subsubsection{Discussions around synchrony and $B$}

The following lemma is key to our construction: it allows us to ensure some pseudo-synchrony in the nodes' awakening schedules.

\begin{lemma} 
\label{lem:Synchrony}
If a node $u$ is commonly awake during round $r$, 
every neighbor of $u$ will be commonly awake at least once between rounds $r+1$ and $r + 2\p_M(\operatorname{id}_u)-1$.
\end{lemma}
\begin{proof} 
Any neighbor $v$ of $u$ is commonly awake with period $\p_M(\operatorname{id}_v)$. 
Since all primes assigned by function $\p_M$ lie in $[ B(M), 2B(M) ]$, we have $\p_M(\operatorname{id}_u) \geq B(M)$ and $\p_M(\operatorname{id}_v) < 2B(M)$, so $\p_M(\operatorname{id}_v) < 2B(M) < 2\p_M(\operatorname{id}_u)$. 
In any interval of $2\p_M(\operatorname{id}_u)$ consecutive rounds, $v$ is awake at least once.
\end{proof}

The following Lemma allows us to argue that $B(M)$ always exists, and gives us some approximation of its value:

\begin{lemma}\label{lem:Bexists}
For any $M$, there exists some $B$ such that the number of prime numbers in $[B,2B]$ is at least $M$.
Moreover, we can find $B$ such as $B=M\log M+o(M\log M)$.
\end{lemma}
\begin{proof}
Let $\pi(x)$ be the number of prime numbers smaller than or equal to $x$. 
By the Prime Number Theorem~\cite{hadamard1896distribution}, which states:
$$\pi(x)=\frac{x}{\log x}+o\left(\frac{x}{\log x}\right)$$
Note that in $[B, 2B]$, we have $\pi(2B)-\pi(B)$ prime numbers. 
In particular, for $B=M^2$, we get:
$\pi(2B)-\pi(B)=\frac{(2M)^2}{\log {(2M)^2}}+o\left(\frac{(2M)^2}{\log (2M)^2}\right)-\frac{M^2}{\log {M^2}}-o\left(\frac{M^2}{\log M^2}\right)=O\left(\frac{M^2}{\log M}\right)$.

It implies that such a $B$ indeed exists and has value at least $M$ (as there are fewer than $M$ prime numbers smaller than $M$).
Let $B(M)$ be the smallest value possible of $B$. 
We have:
$M =\pi(2B(M))-\pi(B(M))= \frac{2B(M)}{\log(2B(M))}-\frac{B(M)}{\log(B(M))}+o\left(\frac{B(M)}{\log(B(M))}\right)$\\
$M =
\frac{B(M)}{\log(B(M))}+o\left(\frac{B(M)}{\log(B(M))}\right)$\\

This equality first implies that $B(M)=M\log(B(M))+o(B(M))$.
Secondly, by taking the logarithm of this equality, we obtain:
$\log M=\log(B(M))+o(\log(B(M)))$ which implies that $\log(B(M))=\log M+o(\log(B(M)))$. This allows us to finally deduce:\\
$B(M)=M\log(B(M))+o(B(M))=M\log M+o(M\log M)$.
\end{proof}

\subsubsection{Safe configurations}

Note that in the algorithm, each node $u$ wakes up every $\p_M(\operatorname{id}_u)$ rounds. They can also wake up right after one of those rounds if they enter in line~\ref{amm:l5} of the algorithm. In the first situation (when $\operatorname{justWokeUp}_u=\operatorname{False}$), we say that $u$ is \emph{commonly awake}, and in the second situation, $u$ is \emph{justly awake}.

\begin{lemma}
\label{lem:mustL}
Let $\gamma_0$ be a configuration in which $\operatorname{mustLeaveMIS}(u)$ evaluates to $\operatorname{False}$ for some node $u$. In any configuration $\gamma_r$ with $r\in\mathbb{N}$, $\operatorname{mustLeaveMIS}(u)$ also evaluates to $\operatorname{False}$.
\end{lemma} 
\begin{proof} 
Suppose, for contradiction, that there exists a configuration $\gamma_r$ in which we have $\operatorname{mustLeaveMIS}(u) = \operatorname{True}$. Let choose $r$ to be minimal, i.e., $\operatorname{mustLeaveMIS}(u) = \operatorname{False}$ in $\gamma_{r - 1}$.

There exists some $v\in N(u)$ such that $\operatorname{status}_u^{r} = \operatorname{status}_v^{r} = \operatorname{inMIS}$. Since $\operatorname{mustLeaveMIS}(u) = \operatorname{False}$ in $\gamma_{r - 1}$, either $u$ or $v$ has changed its status to $\operatorname{inMIS}$ in round $r$.
This means that $u$ or $v$ has applied line~\ref{amm:l3} of the algorithm. However, similarly to Algorithm~\ref{alg:mis}, this can happen only if both $u$ and $v$ have status $\operatorname{outMIS}$ in $\gamma_{r-1}$, and it would imply that both had $\operatorname{awakeAlone}$ verified, which is impossible. This is a contradiction.
\end{proof}

\begin{lemma}
\label{lem:married-stays-married} 
Let $\gamma_0$ be a configuration in which $\operatorname{married}(u)$ evaluates to $\operatorname{True}$ for some node $u$. In any configuration $\gamma_r$ with $r\in\mathbb{N}$, $\operatorname{married}(u)$ also evaluates to $\operatorname{True}$.
\end{lemma}
\begin{proof}
Suppose $\operatorname{married}(u)$ evaluates to True in $\gamma_0$, and suppose for contradiction that there exists a configuration $\gamma_r$ in which $\operatorname{married}(u) = \operatorname{False}$. Let choose $r$ to be minimal, meaning that $\operatorname{married}(u) = \operatorname{True}$ in $\gamma_{r-1}$.

In configuration $\gamma_{r-1}$, we have $\operatorname{status}_u^{r-1} = \operatorname{status}_v^{r-1} = \operatorname{matched}$, $\operatorname{partner}_u^{r-1} = v$ and $\operatorname{partner}_v^{r-1} = u$. Hence, predicates $\operatorname{canJoinMIS}$, $\operatorname{mustLeaveMIS}$, $\operatorname{badMatching}$, $\operatorname{gotMatch}$ and $\operatorname{gotSeduced}$ all evaluate to False. A node can modify its status and its partner only through the rules associated with these predicates, which contradicts $\operatorname{married}(u) = \operatorname{False}$ in configuration $\gamma_r$.
\end{proof}

\begin{lemma} 
\label{lem:badM}
Let $\gamma_0$ be a configuration in which $\operatorname{badMatching}(u)$ evaluates to $\operatorname{False}$ for some node $u$. In any configuration $\gamma_r$ with $r\in\mathbb{N}$, $\operatorname{badMatching}(u)$ also evaluates to $\operatorname{False}$.
\end{lemma}
\begin{proof} 

Suppose $\operatorname{badMatching}(u)$ evaluates to False in $\gamma_0$, and suppose for contradiction that there exists a configuration $\gamma_r$ in which $\operatorname{badMatching}(u) = \operatorname{True}$. Let choose $r$ to be minimal, i.e.,  $\operatorname{badMatching}(u) $ switched to $\operatorname{False}$ in round $r$.
In $\gamma_r$, we have $\operatorname{status}_u = \operatorname{matched}$ and both predicates $\operatorname{married}(u)$ and $\operatorname{waiting}(u)$ evaluate to False.

Since $\operatorname{badMatching}(u) = \operatorname{False}$ in $\gamma_{r-1}$, either $\operatorname{status}^{r-1}_u \neq \operatorname{matched}$, $\operatorname{married}^{r-1}(u) = \operatorname{True}$ or $\operatorname{waiting}^{r-1}(u) = \operatorname{True}$.

Let's suppose $\operatorname{status}^{r-1}_u \neq \operatorname{matched}$.
Then $u$ transitions to $\operatorname{matched}$ in round $r$, which can only happen if either $\operatorname{gotSeduced}(u)$ or $\operatorname{gotMatch}(u)$ is verified.
\begin{itemize}
    \item If $\operatorname{gotSeduced}(u)$ is verified, then in round $r$, $u$ sets $\operatorname{partner}_u$ to some neighbor $v$ satisfying $\operatorname{status}_v = \operatorname{inMIS}$ and $\operatorname{seducing}_v = u$. 
    Hence, $v$ can only update the variable
    $\operatorname{justWokeUp}(v)$ during round $r$, contradicting $\operatorname{waiting}(u) = \operatorname{False}$ in $\gamma_r$.
    \item If $\operatorname{gotMatch}(u)$ is verified, then $u$ sets $\operatorname{partner}_u$ to $v = \operatorname{seducing}_u$ and $\operatorname{status}_u$ to $\operatorname{matched}$ in round $r$. 
    By the definition of $\operatorname{gotMatch}$, the predicate $\operatorname{waiting}(v)$ is verified in $\gamma_{r-1}$. 
    Hence, in round $r$, $v$ can only update the variable
    $\operatorname{justWokeUp}_v$,  contradicting $\operatorname{married}(u) = \operatorname{False}$ in $\gamma_r$.
\end{itemize}

Hence $\operatorname{status}^{r-1}_u = \operatorname{matched}$. 
By Lemma~\ref{lem:married-stays-married}, predicate $\operatorname{married}^{r-1}(u)$ must be False, which implies that  $\operatorname{waiting}^{r-1}(u)=\operatorname{True}$.
Let's call $v = \operatorname{partner}^{r-1}_u$.

Since $\operatorname{status}^{r-1}_u = \operatorname{matched}$, predicates $\operatorname{gotMatch}(u)$, $\operatorname{gotSeduced}(u)$ and $\operatorname{canSeduce}(u)$ evaluate to False in $\gamma_{r-1}$, so $\operatorname{partner}^r_u =v$ and $\operatorname{status}^r_u = \operatorname{matched}$.
So, as predicate $\operatorname{waiting}(u)$ to changes between $\gamma_{r-1}$ and $\gamma_r$, node $v$ must change the values of some its variables in round $r$.

By definition of predicate $\operatorname{waiting}$,
$\operatorname{gotMatch}(v)$ is verified in $\gamma_{r-1}$.
Therefore, if $v$ updates its variables in
round $r$, it will set $\operatorname{partner}_v$ to $\operatorname{seducing}_v = u$, and  $\operatorname{status}_v$ to $\operatorname{matched}$, i. e., $u$ and $v$ will end up married in $\gamma_r$, which contradicts $\operatorname{badMatching}(u)= \operatorname{True}$ in that configuration. 
\end{proof}

\begin{definition}
A configuration is safe if both
predicates $\operatorname{mustLeaveMIS}$ and $\operatorname{badMatching}$ evaluate to $\operatorname{False}$ on any node.
\end{definition}
\begin{corollary}
\label{cor:p(M)}
From any configuration, after the $ \p_M(M)$ rounds, a safe configuration is reached.
The set of safe configurations is closed.
\end{corollary} 
\begin{proof}
Since $\operatorname{id}_u < M$, we have $\p_M(\operatorname{id}_u) \leq \p_M(M)$, so $u$ is awake at least once within the first $\p_M(M)$ rounds. 
If $\operatorname{mustLeaveMIS}(u)$ or $\operatorname{badMatching}(u)$ held at the beginning of this round, $u$ set $\operatorname{status}_u := \operatorname{outMIS}$, making both predicates immediately False ($\operatorname{mustLeaveMIS}$ requires $\operatorname{status}_u = \operatorname{inMIS}$, and $\operatorname{badMatching}$ requires $\operatorname{status}_u = \operatorname{matched}$).
By Lemmas~\ref{lem:mustL} and \ref{lem:badM}, both predicates remain False on every subsequent round.
\end{proof}

\subsubsection{Convergence from a safe configuration}
In this section, we study executions starting from a safe configuration.
Let's define the predicate $\operatorname{isolated}$:
$$\operatorname{isolated}(u) \equiv \operatorname{status}_u = \operatorname{inMIS} \land \forall v \in N(u), \operatorname{status}_v = \operatorname{matched} \land \operatorname{partner}_{v} \neq u$$
We will establish that any execution reaches a configuration where every node is either $\operatorname{married}$ or $\operatorname{isolated}$. We will say (and prove) that such a configuration is \emph{stable}.
\begin{lemma}
\label{lem:closureIsolated}
If in a safe configuration $\gamma$, a node $u$ is isolated, then it stays isolated along any execution from $\gamma$.
\end{lemma}
\begin{proof}
From Corollary~\ref{cor:p(M)}, we know that $\operatorname{mustLeaveMIS}$ and $\operatorname{badMatching}$
are not verified by any node in any configuration reachable from $\gamma$.
We conclude that $u$ and its neighbors will never change their variable $\operatorname{status}$ nor $\operatorname{partner}$.
\end{proof}
\begin{lemma} 
\label{lem:toMatched}
Let $\gamma_0$ be a safe configuration and a node $u$ such that $u$ is commonly awake in round 1, and we have $\operatorname{justWokeUp}^0 = \operatorname{False}$, $\operatorname{status}^0_u = \operatorname{inMIS}$ and $\operatorname{isolated}^0(u) = \operatorname{False}$. 
Let $r=2\p_M(\operatorname{id}_u)+1$. 
There exists some $v\in N(u)$ such that $\operatorname{status}^0_v=\operatorname{outMIS}$ and $\operatorname{status}^r_v=\operatorname{matched}$ (i.e., $v$ changes its status to matched). Moreover, $v$ keeps the status matched in any configuration after $\gamma_r$.
\end{lemma}
\begin{proof} 
Since $\gamma_0$ is safe, $\operatorname{mustLeaveMIS}(u) = \operatorname{False}$, so no neighbor of $u$ has $\operatorname{status} = \operatorname{inMIS}$. 
Since $\operatorname{isolated}(u) = \operatorname{False}$, not every neighbor is matched with a partner other than $u$, so there exists a neighbor $w$ with $\operatorname{status}_w = \operatorname{outMIS}$.

Consequently $\operatorname{canSeduce}(u)$ holds in $\gamma_0$. As $u$ is awake in round 1, there exists some $v\in N(u)$ such that $\operatorname{seducing}^1_u = v$.
By Lemma~\ref{lem:Synchrony}, $v$ is awake during a round $r'\in[2, 2\p_M(\operatorname{id}_u)+1]$. 
At round $r'$, $\operatorname{gotSeduced}(v)$ holds (since $\operatorname{status}_u = \operatorname{inMIS}$ 
and $\operatorname{seducing}_u = v$). 
Moreover, since $\gamma_0$ is safe, we have $\operatorname{mustLeaveMIS}(v) = \operatorname{badMatching}(v) = \operatorname{False}$, so $v$ executes lines~\ref{amm:l13} and \ref{amm:l14} of the algorithm, setting $\operatorname{partner}_v$ to one of its neighbors (possibly $u$) and $\operatorname{status}_v$ to $\operatorname{matched}$.
As $\gamma_0$ is safe, $\operatorname{status}_v$
stays unchanged after the round $r'$.
\end{proof}

\begin{corollary} 
\label{cor:inMIS}
Let $\gamma_0$ be a safe configuration and a node $u$ such that $\operatorname{status}_u = \operatorname{inMIS}$. Then $u$ is either 
$\operatorname{married}$ or $\operatorname{isolated}$ in any configuration $r$ with $r\ge (2\Delta+2)\p_M(\operatorname{id}_u)+\Delta$.
\end{corollary}
\begin{proof} 
As $\gamma_0$ is safe, a node $\operatorname{married}$
(resp. $\operatorname{isolated}$) stays $\operatorname{married}$
(resp. $\operatorname{isolated}$).
If $\operatorname{isolated}(u)$ already holds in $\gamma_0$, the claim is immediate.

Otherwise, we apply Lemma~\ref{lem:toMatched} inductively. 
As long as $u$ remains $\operatorname{inMIS}$ and not isolated, every $2\p_M(\operatorname{id}_u)+1$ rounds (after the first time $u$ was commonly awake), some
neighbor of $u$ sets its status from $\operatorname{outMIS}$ to $\operatorname{matched}$, and keeps this status forever.

Since $u$ has at most $\Delta$ neighbors, after $(2\Delta+1)\p_M(\operatorname{id}_u)+\Delta$ rounds, either all neighbors of $u$ are matched with partner $\neq u$, i.e., $\operatorname{isolated}(u) = \operatorname{True}$, or some neighbor $v$ has $\operatorname{status}_v = \operatorname{matched}$ and $\operatorname{partner}_v = u$, so $\operatorname{gotMatch}(u)$ is verified and $u$ gets married with $v$ during the next awake round.
\end{proof}

\begin{lemma} 
\label{lem:quitOutMS}
Let $\gamma_0$ be a safe configuration and a node $u$ commonly awake in round 1 such that $\operatorname{status}^0_u = \operatorname{outMIS}$, and for each neighbor $v$ of $u$, we have $\operatorname{status}^0_v \neq \operatorname{inMIS}$.

Let $F$ be the set of neighbors of $u$ with $\operatorname{status} = \operatorname{outMIS}$ in $\gamma_0$.
In $\gamma_r$, with $r=2+ \Delta\p_M(\operatorname{id}_u)$, either $u$ or some node in $F$ has a different status than in $\gamma_0$.
\end{lemma}
\begin{proof} 
This proof is similar to that of Lemma~\ref{lem:activelyawake} for the MIS algorithm. Since $\operatorname{status}_u = \operatorname{outMIS}$ and no neighbor has $\operatorname{status} = \operatorname{inMIS}$, predicate $\operatorname{canJoinMIS}(u)$ holds in $\gamma_0$. 
Suppose, for contradiction, that $u$ and all nodes in $F$ still have status $\operatorname{outMIS}$ in $\gamma_r$. 
Since $\gamma_0$ is safe, no node can change its status
to $\operatorname{outMIS}$; this implies that the status of every node in $F$ is unchanged. 
So $\operatorname{canJoinMIS}(u)$ holds from rounds 1 to $r$.

Before round $r$, $u$ will be justly awake (i.e., awake with $\operatorname{justWokeUp}_u = \operatorname{True}$)  $\Delta+1$ times at rounds $2+i\p_M(\operatorname{id}_u)$ with $0 \leq i \leq \Delta$.
If $\operatorname{awakeAlone}(u)$ holds at any of these $\Delta + 1$ rounds, $u$ joins the MIS, contradicting our assumption.
Hence $\operatorname{awakeAlone}(u) = \operatorname{False}$ every time $u$ is justly awake, i. e., some neighbor of $u$ is also justly awake. 
Since $u$ has at most $\Delta$ neighbors, by the pigeonhole principle, some neighbor $v$ is justly awake simultaneously with $u$ two times. 

Let's call $r_1$ the first round in which both $u$ and $v$ are justly awake at the same time, and $r_2$ the next. Note that $r_1 \leq r_2 \leq r_1 + \Delta  \p_M(\operatorname{id}_u) $.

Note that, by construction of Algorithm~\ref{alg:mm}, if a node $w$ is justly awake in two different rounds $s_1$ and $s_2$, we have that $s_1-s_2$ is a multiple of $\p_M(\operatorname{id}_w)$.

We get that $r_2-r_1$ is both a multiple of $\p_M(\operatorname{id}_u)$ and $\p_M(\operatorname{id}_v)$, which are different prime numbers. Hence, $r_2-r_1=k\p_M(\operatorname{id}_u)\p_M(\operatorname{id}_v)$ for some $k\ge1$. However, we have $\p_M(\operatorname{id}_v)>\Delta$, which brings a contradiction to $r_2 \leq r_1 + \Delta  \p_M(\operatorname{id}_u)$. 
\end{proof}

\begin{lemma} 
\label{lem:outMIS}
Let $\gamma_0$ be a safe configuration and a node $u$ such that $\operatorname{status}_u = \operatorname{outMIS}$ in $\gamma_0$. Let $r=2+(5\Delta+4)\p_M(\operatorname{id}_u)+\Delta$.
In $\gamma_{r}$, either $\operatorname{isolated}(u)$ holds or the number of matched nodes in $N[u]$ has strictly increased.
\end{lemma}
\begin{proof}
Assume that, in $\gamma_r$, $u$ is not isolated and that the number of matched nodes in $N[u]$ has not increased. Since $\gamma_0$ is safe, no matched node in $\gamma_0$ has changed its status in $\gamma_r$. Hence, no node in $N[u]$ takes the status $\operatorname{matched}$ before $\gamma_r$.

Suppose some neighbor $v$ has $\operatorname{status}_v = \operatorname{inMIS}$ in $\gamma_0$.
By applying Corollary~\ref{cor:inMIS}, $v$ is either isolated or married in $\gamma_r$.
\begin{itemize}
	\item If $v$ is married: $\operatorname{status}_v = \operatorname{matched}$, so the number of matched nodes has increased. Contradiction.
	\item If $v$ is isolated: every neighbor of $v$ (including $u$) has $\operatorname{status} = \operatorname{matched}$ with a partner $\neq v$. In particular, $u$ is matched, so the number of matched nodes has increased. Contradiction.
\end{itemize}
Hence, no neighbor of $u$ has $\operatorname{status} = \operatorname{inMIS}$  in $\gamma_0$.

By Lemma~\ref{lem:quitOutMS}, by round $r_1=2+ \Delta \p_M(\operatorname{id}_u)$, a node $v\in N[u]$ changed its status from $\operatorname{outMIS}$. 
If $v$ has taken the status $\operatorname{matched}$, then the number of matched nodes has increased.
Hence, $v$ has taken the status $\operatorname{inMIS}$.
By applying Corollary~\ref{cor:inMIS} to $v$, we reach the same contradictions as previously after $(2\Delta+2)\p_M(\operatorname{id}_v)+\Delta$ more rounds. 
As $\p_M(\operatorname{id}_v) < 2\p_M(\operatorname{id}_u)$,
in round $r > 2+ \Delta \p_M(\operatorname{id}_u)+(2\Delta+2)\p_M(\operatorname{id}_v) + \Delta$, the number of matched nodes must have increased.
\end{proof}

\begin{corollary}\label{cor:term}
After at most $r=\p_M(M)[\Delta(5\Delta+4)+1]+\Delta(\Delta+2) $ rounds, any node is either isolated or married.
\end{corollary}
\begin{proof}
After $\p_M(M)$ rounds, by Corollary~\ref{cor:p(M)}, a safe configuration has been reached.
Let us assume that in $\gamma_{\p_M(M)}$, there exists some node $u$ that is neither isolated nor married. 
If $\operatorname{status}_u=\operatorname{inMIS}$, by Corollary~\ref{cor:inMIS}, $u$ is isolated or married in $\gamma_r$.
Otherwise, we can apply Lemma~\ref{lem:outMIS} iteratively (at most $\Delta$ times) until either:
\begin{itemize}
    \item $u$ gets matched
    \item All the neighbors of $u$ get matched.
\end{itemize}
In both scenarios, $u$ ends up married or isolated, and this happens before $\gamma_r$, as $\p_M(\operatorname{id}_u)\le\p_M(M)$.
\end{proof}

\subsubsection{Stable configurations}

A configuration is said to be stable if every node is either isolated or married.
\begin{lemma}\label{lem:stable-closure}
The set of stable configurations is closed.
\end{lemma}
\begin{proof}
In a  stable configuration, the predicates $\operatorname{canJoinMIS}$, $\operatorname{mustLeaveMIS}$, $\operatorname{badMatching}$, $\operatorname{gotMatch}$, 
$\operatorname{gotSeduced}$ and $\operatorname{canSeduce}$ evaluate at $\operatorname{False}$ on every node.
\end{proof}

\begin{theorem}\label{theo:stable-mm}
In a stable configuration, the set of edges $\{u, v\}$ such that $\operatorname{status}_u = \operatorname{status}_v = \operatorname{matched}$, $\operatorname{partner}_u = v$, and $\operatorname{partner}_v = u$ defines a maximal matching.
\end{theorem}
\begin{proof}
Suppose that it is possible to add one more edge $\{ u, v \}$ to the matching so that it remains a legal matching. Hence, both $u$ and $v$ are not matched yet so $\operatorname{married}(u) = \operatorname{married}(v) = \operatorname{False}$. Since the configuration is terminal, predicates $\operatorname{isolated}(u)$ and $\operatorname{isolated}(v)$ must evaluate to $\operatorname{True}$. But by definition, two neighbors cannot simultaneously verify the predicate isolated. Contradiction. 
\end{proof}

\begin{theorem}
Algorithm~\ref{alg:mm} computes a Maximal Matching in 
$O(\Delta^2 M\log M)$ rounds. 
Each node is awake at most $O(\Delta^2)$ times
during the convergence phase. 
The awake frequency after convergence is $O(1/(M\log M))$.
\end{theorem}
\begin{proof}

A stable configuration is reached after 
$O( \p_M(M)\Delta^2)$ rounds according to corollary~\ref{cor:term}.
In a stable configuration, a Maximal Matching is computed, thanks to Lemma~\ref{lem:stable-closure} and Theorem~\ref{theo:stable-mm}, forever.
By definition, $p(M) \leq 2B$. We have $B = O(M \log M)$ (Lemma~\ref{lem:Bexists}).
By the choice of $B$, a node is awake at least once and at most four times in $\p_M(M)$ rounds, so the awake complexity of a node before convergence of the system is at most $ \operatorname{O} (\Delta^2)$, and the awake frequency after convergence, for each node, is $O(1/\p_M(M))$.
\end{proof}

\section{Two Automatic Transformers}\label{sec:trans}

In this section, we introduce two ways to transform a self-stabilizing algorithm to make it more energy-efficient by lowering the frequency with which a node is awake. Algorithm~\ref{alg-sleeping-simple} simply forces each node to be awake once every $SP$ rounds, for some constant $SP$. Algorithm~\ref{alg-sleeping-fast} is more elaborate and tries to converge faster when a node detects that some computation is required. In this section, we will denote the initial algorithm by $\InAlg$. 

\subsection {A Simple Transformer: Algorithm~\ref{alg-sleeping-simple}}
Algorithm~\ref{alg-sleeping-simple} is a simple transformer that converts a self-stabilizing algorithm into a self-stabilizing SLEEPING Algorithm. 
The transformed algorithm, denoted by $\InAlg$, only needs to converge under the weakly fair scheduler.
If $\InAlg$ converges in $T$ rounds then the obtained algorithm converges
in $T.SP$ synchronous rounds, for any arbitrary $SP$.
The wake-up frequency is always the same $1/SP$.
Therefore, during the convergence time, each node is awake during $T$ synchronous rounds.
\begin{algorithm}
$~$
\BlankLine
 \textbf{Inputs:} \\
 \quad $\InAlg$ :  a self-stabilizing algorithm under the weakly fair scheduler\\
 \quad $\text{SP}$ :  an integer being the awake frequency\\
 
$~$
\BlankLine
\SetKwBlock{Begin}{Upon waking:}{} 
\Begin{
\nl $u$ executes an action of $\InAlg$ if it is enabled;
 Sleep for $SP$ rounds;\\
}

$~$
\BlankLine
\caption{Simple transformer to make a self-stabilizing algorithm energy efficient}
\label{alg-sleeping-simple}
\end{algorithm}
\subsection {A More Elaborate Transformer: Algorithm~\ref{alg-sleeping-fast}}
Algorithm~\ref{alg-sleeping-fast} is a transformer that allows converting a silent self-stabilizing algorithm under the synchronous scheduler into a silent self-stabilizing SLEEPING Algorithm. 
If $\InAlg$ converges in $T$ rounds under the synchronous scheduler, then the obtained algorithm converges in $(3D+2).SP+3T$ synchronous rounds, where $D$ is the diameter of the graph. 

The idea of the algorithm is to (1) detect that a node has not converged, (2) broadcast the information to the whole graph, (3) have $T$ synchronous rounds where every node is active to converge, and (4) go back to sleeping with a $1/SP$ awake frequency (where every node actually wake up at the same time every $SP$ rounds).

More formally, the first node, $u$, to locally detect that $\InAlg$ has diverged stays awake for $D.SP+T$ consecutive rounds. More specifically, $u$ sets its variable $ect$ (elapsed convergence time) to $1$. The first period of $D.SP$ rounds is used to wake up all nodes.
In the second period of $T$ rounds, every node executes the algorithm $\InAlg$.
By the definition of $T$, a terminal configuration of $\InAlg$ is reached at the end of this period.
During the final period of at most $D.SP+T$ rounds, every node ends up setting its $ect$ variable to $0$ to reach a silent configuration.

\begin{algorithm}[htb]
$~$
\BlankLine
 \textbf{Inputs:} \\
 \quad $\InAlg$ :  a self-stabilizing algorithm under the weakly fair scheduler\\
 \quad $\text{SP}$ : an integer being the awake frequency after convergence\\
 \quad $\text{D}$ : an upper bound on the diameter of the graph\\
  \quad $\text{T}$ : an upper bound on the maximum number of rounds needed for the convergence of $\InAlg$\\

\BlankLine
 \textbf{Local variable:} \\
    \quad $\sleepingDuration_u$ : integer taking value in $\{1, SP\}$\\

\BlankLine
 \textbf{Shared variable:} \\
    \quad $\countVar_u$ : integer taking value in $[0, D.SP+T]$\\
    
\BlankLine
 \textbf{macro "Elapsed Convergence Time":} 
 $$\maxCount_u := \left\{
    \begin{array}{ll}
        0 & \mbox{if  } \forall v \in N[u], \countVar_v = 0 \\
        1+ \text{min of }\{~\countVar_v ~|~v \in N[u] \text{ and } \countVar_v > 0 \} & \mbox{otherwise}
    \end{array}
\right.$$

\BlankLine
\SetKwBlock{Begin}{Upon waking:}{} 
\Begin{
\tcc{Regularly wake up to detect an anomaly.}
\nl \lIf 
{$u$ is disabled in $\InAlg$ and $\maxCount_u = 0$} 
{ 
$\sleepingDuration_u := SP$ \label{a2:l1}
}
\nl \Else  {
\nl \lIf  {$u$ is enabled in $\mathcal A$} 
{
$u$ executes an action of $\mathcal A$
}
\nl \If {$\maxCount_u < D.SP+T$}
 { 
\nl $\countVar_u := \max(1,\maxCount_u$); \label{a2:l5}
 $\sleepingDuration_u := 1$;
} 
\nl \lElse { \label{a2:l6}
 $\countVar_u := 0$; $\sleepingDuration_u := SP$  } }
\nl Sleep for $\sleepingDuration_u$ rounds.
 }

$~$
\BlankLine
\caption{Elaborate transformer to make a self-stabilizing algorithm energy efficient}
\label{alg-sleeping-fast}
\end{algorithm}

To prove that our transformer works, we first need the following predicate that ensures that nodes are sleeping only if their variable $\countVar$ is set to 0, and prove that this predicate is and remains true after most $SP$ rounds.

$$\text{Predicate }\pseudoCorrect(u): \sleepingDuration_u = SP \Leftrightarrow \countVar_u= 0$$
\begin{lemma}
\label{lem:pseudo-correct}
After $SP$ rounds, every node verifies the predicate $\pseudoCorrect$.
The predicate $\pseudoCorrect$ is closed.
\end{lemma}
\begin{proof}
The variable $\sleepingDuration$ takes values in $\{1,SP\}$, meaning that every node is active at least once every $SP$ rounds.
An awake node updates its two variables $\sleepingDuration$ and $\countVar$ in such a way that
it verifies the predicate $\pseudoCorrect$ at the end of the current round.
\end{proof}

We introduce, for any configuration $\gamma$, $\ectMinStrictPositif_{\gamma}$ as the minimum value of $\countVar$ among all the nodes. Formally,
$$\ectMinStrictPositif_{\gamma}=
\left\{
    \begin{array}{ll}
        0 & \mbox{if  } \forall u \in V, \countVar_u = 0 \\
        \text{min of }\{~\countVar_u ~|~ \countVar_u > 0 \} & \mbox{otherwise}
    \end{array}
\right.
$$

\subsubsection{Executions from a pseudoCorrect configuration}
A configuration is said to be pseudoCorrect if all nodes verify the predicate $\pseudoCorrect$. In this section, we show that from such a configuration, a terminal configuration of $\InAlg$ is reached.

\begin{lemma}
\label{lem:ronde-sync-v2}
Let $\gamma_0$ be a pseudoCorrect configuration.
Let us assume that there exists some node $u$ such that $\countVar_u^0 \le (D-1)SP$.

In round $\gamma_{SP}$, 
we have $\countVar_v^{SP} \in [1,\countVar_u^0+SP]$ 
if $v$ is $u$ or a neighbor of $u$.
\end{lemma}
\begin{proof}
As long as $\countVar_u < D.SP+T$, the value $\countVar_u$ can only increase by at most one by Algorithm~\ref{alg-sleeping-fast} before reaching (if ever) the value 0. Hence, for any configuration $\gamma_l$  with $l\in[1, SP]$, we get that $1\le \countVar_u^l\le \countVar_u^0+l$.

Let $v\in N(u)$. 
There exists $i\in[0,SP-1]$  such that $v$ is awake in configuration $\gamma_i$. 
We have $\maxCount_v^i\in[2,\countVar_u^i+1]$. 
Hence, $\countVar_v^{i+1}=\maxCount_v^i$. 
As $\maxCount_{v}$ remains positive and increases by at most $1$ during $[i, SP]$, 
we can conclude
that $\countVar_v^{SP}\in[1,\countVar_u^0+SP]$.
\end{proof}
\begin{lemma}
\label{lem:ronde-sync}
Let $\gamma_0$ be a pseudoCorrect configuration.
Let us assume that there exists some node $u$ such that $\countVar_u=1$. 
All nodes are active during the $T$ consecutive rounds from $\gamma_{D.SP}$ to $\gamma_{D.SP+T}$ 
and $\gamma_{D.SP+T}$ is a terminal configuration of $\InAlg$.
\end{lemma}

\begin{proof}
First, we observe that, if $\countVar_v=k>0$ for some node $v$, then this value increases by at most 1 before reaching $D.SP+T$. Hence, $v$ will be awake for at least $D.SP+T-k$ consecutive rounds.

Let's prove by induction on $d$ that nodes at distance $d$ from $u$ are active in configurations $\gamma_{d.SP}$ to $\gamma_{D.SP+T}$.
For $d=0$, it is immediate, as $u$ is the only node at distance 0 from $u$. By applying Lemma~\ref{lem:ronde-sync-v2}, we immediately get the inductive step.

Since all nodes are at a distance of at most $D$ from $u$, all nodes are active in round $D.SP$ with their $\countVar$ variable at values in $[1,D.SP]$. From this round, they will all be active for the next $T$ consecutive rounds.
\end{proof}
\begin{lemma}
\label{lem:strict-crois-B}
Let $\gamma_0$ be a pseudoCorrect configuration of Algorithm~\ref{alg-sleeping-fast}.
Let $u$ be a node.
Assume that  $\countVar_u^1\neq 1$.
After a round from $\gamma_0$,
 either $\countVar_u^1=0$ or the value of $\countVar_u^1 > \ectMinStrictPositif^0$.
\end{lemma}
\begin{proof}
We study the behavior of $u$ during the round from $\gamma_0$.
Assume that $u$ executes line~\ref{a2:l5} during the round from $\gamma_0$.
As  $\maxCount_u^1 \neq 0$ we have $\maxCount_u^1 \geq \ectMinStrictPositif^{0}+1$.
Hence, the value of $\countVar_u^1 > \ectMinStrictPositif^{0}$.

Assume that $u$ executes line~\ref{a2:l6} during the round from $\gamma_0$.
After the round,  $\countVar_u = 0$.

Assume that $u$ is awake and does not execute line~\ref{a2:l5} or~\ref{a2:l6} during the round from $\gamma_0$. 
So,  $\maxCount_u^0=0$ and the value of $\countVar_u$ is unchanged during this round. 

Assume that $u$ is not awake. 
As the predicate $\pseudoCorrect$ is verified in $\gamma_0$;
after the round, $\countVar_u = 0$.
\end{proof}
$$\text{Predicate }\correct(u): \pseudoCorrect(u)~\&~ \countVar_u= 0$$
\begin{lemma}
\label{lem:ronde-cons}
Let $\epsilon$ be an execution from $\gamma_0$, a pseudoCorrect configuration.
Assume during the first $D.SP +T$ configurations of $\epsilon$, 
there is not  node $u$  such that 
$\countVar_u=1$. 
There is a configuration  among the first $D.SP +T$ configurations of $\epsilon$
such that every node verifies the predicate $\correct$. 
\end{lemma}
\begin{proof}
Any reached configuration of $\epsilon$ is $\pseudoCorrect$ (Lemma~\ref{lem:pseudo-correct}).
So it is enough to establish that
a configuration where every $\countVar$ variable has the value $0$
is reached.

If in  $\gamma_0$, every node $u$ verifies $\countVar_u= 0$ then 
$\gamma_0$ is the requested configuration.

Assume that in $\gamma_0$ a node verifies $\countVar_u> 0$.
So, $\ectMinStrictPositif^{0} \geq 2$ by hypothesis.
We assume that  during the first $D.SP +T$ configurations of $\epsilon$, there is not  node $u$  such that $\countVar_u=1$. 
So, the repeated use of lemma~\ref{lem:strict-crois-B}
establishes that during the 
 first $D.SP +T-\ectMinStrictPositif^{0}+1$ rounds, 
 the value $\ectMinStrictPositif$ increases to reach $D.SP +T+1$.
After $D.SP +T-\ectMinStrictPositif^{0}+1$ rounds,
 on every node, $\countVar$ has the value $0$.
We conclude that every node verifies the predicate $\correct$.
\end{proof}

\begin{lemma}
\label{lem:ronde-T}
Let $\epsilon$ be an execution from $\gamma_0$ a $\pseudoCorrect$ configuration.
After $(D+1).SP +T$ rounds, a terminal configuration of $\InAlg$ is reached or 
 there is a node $u$ such that  $\countVar_u=1$ during one of the first $(D+1).SP+T$ configurations of $\epsilon$. 
\end{lemma}
\begin{proof}
Assume that no node $u$ verifies $\countVar_u =1$ during 
the first $D.SP +T$ configurations of $\epsilon$. 
According to Lemma~\ref{lem:ronde-cons}, there is a configuration  among the first $D.SP +T$ configurations of $\epsilon$
such that every node verifies the predicate $\correct$. 
Let $\gamma$ be that configuration.

Assume that $\gamma$ is not a terminal configuration of $\InAlg$.
Let $r$ be  the first round along $\epsilon$ from $\gamma$ 
where an enabled node for $\InAlg$, $u$, is awake.
We get that $\countVar_u^r = 1$: $u$ has executed line~\ref{a2:l5} when all nodes verify the predicate $\correct$.
Since each node wakes up with a frequency greater than or equal to 1/SP, $r \leq SP$.
\end{proof}

\subsubsection{Executions from a terminal pseudoCorrect configuration}

We say that a configuration is \emph{legitimate} when every node $u$ verifies the predicate $\correct(u)$ 
and $u$ is disabled for $\InAlg$. We show in this section that a legitimate configuration is always reached.

\begin{lemma}
\label{lem:strict-crois}
Let $\gamma_0$ be a $\pseudoCorrect$ configuration of Algorithm~\ref{alg-sleeping-fast} 
where a terminal configuration of $\InAlg$ is reached.
Let $u$ be a node.
After a round from $\gamma_0$,
 either $\countVar_u=0$ or the value of $\countVar_u > \ectMinStrictPositif^{0}$.
\end{lemma}
\begin{proof}
Let $u$ be a node. We study the behavior of $u$ during the round from $\gamma_0$.
Assume that $u$ executes line~\ref{a2:l5} during the round from $\gamma_0$.
By hypothesis, in $\gamma_0$, $u$ is not enabled for $\InAlg$; 
so $\maxCount_u \neq 0$ in $\gamma_0$. 
Thus $\maxCount_u \geq \ectMinStrictPositif^{0}+1$. 
Hence, the value of $\countVar_u > \ectMinStrictPositif^{0}$, 
after the  round from $\gamma_0$.

Assume that $u$ executes line~\ref{a2:l6} during the round from $\gamma_0$.
After the round,  $\countVar_u = 0$.

Assume that $u$ is awake and does not execute line~\ref{a2:l5} or~\ref{a2:l6} during the round from $\gamma_0$. 
So,  $\maxCount_u=0$ in $\gamma_0$ and the value of $\countVar_u$ is unchanged during this round. 
Therefore, the value of $\countVar_u$ is still zero after the round. 

Assume that $u$ is not awake. 
As the predicate $\pseudoCorrect$ is verified in $\gamma_0$;
after the round,  $\countVar_u = 0$.
\end{proof}

\begin{corollary}
\label{cor:strict-crois}
Let $\gamma_0$ a $\pseudoCorrect$ configuration of Algorithm~\ref{alg-sleeping-fast} 
where a terminal configuration of $\InAlg$ is reached.
After at most $D.SP+T$ rounds, a legitimate configuration is reached.
The set of legitimate configurations is closed.
\end{corollary}
\begin{proof}
Let $\gamma_0$ be a configuration.
If $\ectMinStrictPositif^{0} = 0$ then $\gamma_0$ is a legitimate configuration - 
every node $u$ verifies $\countVar_u =0$ and $\sleepingDuration_u := SP$.
Only line~\ref{a2:l1} is executed by the awake nodes. 
So, the reached configuration is legitimate.

Assume that $\ectMinStrictPositif^{0} > 0$.
Let $\gamma_1$ be the configuration reached from $\gamma_0$.
According to lemma~\ref{lem:strict-crois}, 
$\ectMinStrictPositif^{1} > \ectMinStrictPositif^{0}$ or $\gamma_1$ is a legitimate configuration. 
As $\ectMinStrictPositif^{\gamma}$ is bounded by $D.SP+T$. 
After $D.SP+T - \ectMinStrictPositif^{0}$ rounds, 
every node $u$ verifies $\countVar_u = 0$
and $\sleepingDuration_u := SP$.
\end{proof}

\begin{theorem}
Let $\epsilon= \gamma_0, \gamma_1, \cdots$ be an execution of Algorithm~\ref{alg-sleeping-fast}.
After $3(D.SP +T)+2SP$ rounds, a legitimate configuration is reached.
\end{theorem}
\begin{proof}
From $\gamma_0$, a configuration $\gamma_r$ where every node verifies the predicate $\pseudoCorrect$ is reached after at most $SP$ rounds.

First case. There is a node $u$ such that $\countVar_u =1$ at $\gamma_{r'}$
where $\gamma_{r'}$ is one of the first $(D+1).SP +T$ configurations of $e$ after reaching $\gamma_{r}$.
According to Lemma~\ref{lem:ronde-sync}, 
$\gamma_{r''}$ with $r'' = r' + D.SP+T \leq 2(D.SP +T +SP) $ is a terminal configuration of $\InAlg$.

Second case. No node $u$ verifies $\countVar_u =1$ during 
the first $(D+1).SP +T$ configurations of $e$ after reaching $\gamma_{r}$.
According to Lemma~\ref{lem:ronde-T}, 
$\gamma_{r'}$ is a terminal configuration of $\InAlg$ with $r' = r +(D+1)SP+T \leq 
D.SP +T +2SP$.

To conclude, after $s= 2(D.SP+SP+T)$ rounds, 
a terminal configuration of $\InAlg$ is reached. 
According to Corollary~\ref{cor:strict-crois},
$\gamma_{s'}$ with $s' = s + D.SP+T \leq 3(D.SP +T)+SP$ is a silent legitimate configuration.
\end{proof}

\section{Conclusion}

In this paper, we have introduced a new model that incorporates energy efficiency into self-stabilization by allowing nodes to be inactive. We have introduced several algorithms to solve classical problems, with new techniques that revolve in particular around the use of prime numbers to simulate asynchrony, mutual exclusion, and pseudo-synchrony. We have also implemented two transformers to automatically make any self-stabilizing algorithm more energy-efficient, at the cost of some time-complexity overhead.

Upcoming work can focus on developing new techniques to improve the complexities of the algorithms we provide and solve other problems, such as leader election or spanning trees.

\newpage


\end{document}